\documentclass[a4paper,UKenglish]{lipics-v2021}

\usepackage{microtype}
\usepackage{algorithm,algorithmicx,algpseudocode}
\usepackage{cleveref}
\usepackage{makecell}

\long\def\longdelete#1{}

\newcommand{\costalg}{|\textsc{SSOP}|}
\newcommand{\costopt}{|\textsc{OPT}|}

\newcommand{\approxNT}{\rho}

\newcommand{\smartstart}{\textsc{SmartStart}}
\newcommand{\ouralg}{\textsc{SSOP}}

\newcommand{\pah}{\textsc{PAH}}
\newcommand{\ignore}{\textsc{Ignore}}
\newcommand{\replan}{\textsc{Replan}}

\newcommand{\delaytrust}{\textsc{DelayTrust}}

\algrenewcommand\algorithmicrequire{\textbf{Input:}}
\algrenewcommand\algorithmicensure{\textbf{Output:}}

\usepackage{thm-restate}

\title{Waiting is worth it and can be improved with predictions}

\author{Ya-Chun Liang}{
Department of Computer Science, National Tsing Hua University, Hsinchu, Taiwan}{liangyc@gs.ncku.edu.tw}{}{}

\author{Meng-Hsi Li}{Department of Industrial Engineering and Engineering Management,
National Tsing Hua University, Hsinchu, Taiwan}{ted870924@gapp.nthu.edu.tw}{}{}

\author{Chung-Shou Liao}{Department of Industrial Engineering and Engineering Management,
National Tsing Hua University, Hsinchu, Taiwan \\ 
Department of Electrical Engineering,
National Taiwan University, Taipei, Taiwan}{csliao@ie.nthu.edu.tw}{}{}

\author{Clifford Stein}{Department of Industrial Engineering and Operations Research, Columbia University, New York, USA}{cliff@ieor.columbia.edu}{}{}

\authorrunning{Y.-C. Liang, M.-H. Li, C.-S. Liao, and C. Stein}



\nolinenumbers 

\begin{document}

\maketitle



\begin{abstract}
We revisit the well-known online traveling salesman problem (OLTSP) and 
its extension, the online dial-a-ride problem (OLDARP).
A server 
starting at a designated origin in a metric space, is required to serve online 
requests, 
and return to the origin such that the completion time is minimized. The \smartstart~algorithm, introduced by Ascheuer et al., incorporates a waiting approach into an online schedule-based algorithm and attains the 
optimal 
upper bound of
2 for 
the OLTSP and
the OLDARP if each 
schedule
is optimal. 
Using the Christofides’ 
heuristic
to approximate each schedule leads to the 
currently best upper bound of 
$\frac{7+\sqrt{13}}{4} \approx$ 2.6514 in polynomial time.

In this study, we investigate how an online algorithm with predictions, a recent popular framework (i.e. the so-called learning-augmented algorithms), can be used to improve the best competitive ratio
in polynomial time. 
In particular, we develop a waiting strategy with \emph{online predictions}, 
each of which is only a binary decision-making for every schedule in a whole route,
rather than forecasting an entire set of requests in the beginning (i.e. \emph{offline predictions}). 
That is, it does not require knowing the number of requests in advance.
The proposed online schedule-based algorithm can 
achieve $(1.1514\lambda+1.5)$-consistency and $(1.5+\frac{1.5}{2.3028\lambda-1})$-robustness 
in polynomial time,   
where $\lambda \in (\frac{1}{\theta},1]$ 
and $\theta$ is set to 
$\frac{1+\sqrt{13}}{2} \approx$ 2.3028.
The best consistency 
tends to approach to~2 when $\lambda$ 
is close to
$\frac{1}{\theta}$.  
Meanwhile, 
we show any online schedule-based algorithms cannot derive a competitive ratio of less than 2 even with perfect online predictions. 

\end{abstract}


\section{Introduction}


In many real-world scenarios,  
input information is not initially known but 
gradually becomes clear over time. For instance, 
there are always uncertainties  
during various stages of the manufacturing process on a production line.  There are numerous unforeseen events that cannot be predicted well, 
like: 
unexpected equipment failures, material shortages, or variations in production demands. To tackle such problems, we often 
exploit statistical 
learning techniques or design online algorithms to provide 
an approximation
solution. 

In the field of machine learning, a learning model utilizes the analysis of past acquired information to make predictions about the future. However, 
it usually lacks theoretical guarantees to ensure the quality of its predictions. On the other hand, online algorithms provide a theoretical evaluation using a competitive ratio, which is defined to be the ratio of the optimal offline cost to the worst-case one of an online algorithm. The worst-case competitive analysis of an online algorithm is sometimes too pessimistic, in which worst cases seldom occur in real-world applications 
though.

\emph{A learning-augmented algorithm} combines the two strategies and provides a new way; precisely, it 
involves incorporating a prediction strategy into an online algorithm. 
Unlike pure online algorithms, learning-augmented algorithms further discuss the \emph{consistency} and \emph{robustness}, rather than only competitive ratios.
Briefly speaking, a learning augmented algorithm achieves good performance when predictions are assumed to be perfectly accurate, which is called \emph{consistency}. On the other hand, if the predictions are terribly poor, the algorithm can still maintain a theoretical bound in the worst case, which is called \emph{robustness}.

In recent years, there has been a rapid growth of employing different predictive strategies to design learning-augmented algorithms for a variety of online optimization problems. These problems span various domains, including ski-rental \cite{NEURIPS2021_Antoniadis_ski, kumar2018improving_ski,NEURIPS2020_Wang_ski, wei2020optimal_ski}, scheduling \cite{SODA22_Sara_scheduling,ijcai2022p636_scheduling}, bidding \cite{NIPS2017_Munoz_bidding}, knapsack \cite{NEURIPS2021_Im_knapsack}, covering \cite{NIPS20_Bamas_covering, NEURIPS2022_Grigorescu_covering}, paging \cite{SODA22_Nikhil_paging,SODA20_Dhruv_paging}, matching \cite{NEURIPS2022_Dinitz_matching,NEURIPS2022_Jin_matching}, bin packing \cite{angelopoulos2018binpacking,ijcai2022p635_bin} and graph algorithmic problems \cite{SODA22_Yossi_graph,ICML22_Chen_graph}.

In this study, 
we discuss the well-known online traveling salesman problem (OLTSP) and its
extension,
the online dial-a-ride problem (OLDARP).
For the OLTSP, 
Ausiello~et~al.~\cite{ausiello2001algorithms} proposed a 2-competitive algorithm for the OLTSP, called \pah, if each offline TSP route is optimal. It achieves 3-competitive in polynomial time when using the Christofides' heuristic. Moreover, there are other types of online algorithms including \replan, \ignore, and \smartstart, introduced in~\cite{ascheuer2000online} by Ascheuer~et~al. 
Note that the latter two belong to the category of \emph{schedule-based} algorithms~\cite{phdthesis2020,zhang}.
A schedule-based algorithm enables the server to follow a predetermined (i.e. offline) walk during each schedule and ignore any new incoming requests that arrive while performing the current schedule. The deferred requests are incorporated into the construction of the next schedule.
In other words, the solution produced by a schedule-based algorithm may comprise one or more such schedules. The formal definition of schedule-based algorithms will be presented in the next section. 
In particular, the \smartstart~algorithm has the currently best upper bound of 2.6514 for the OLTSP and the OLDARP,
when using the Christofides' heuristic to approximate each schedule.  
Readers may refer to a summary of all these online algorithms for the OLTSP and the OLDARP~\cite{phdthesis2020} by Birx.

For the related works of online routing problems under the learning-augmented framework, 
Bernardini et al.~\cite{bernardini2022a} focused on defining a universal error measure which can be used for different problems, including the OLTSP and the OLDARP. They predicted a sequence of requests and applied a graph network to compute the prediction errors between the predicted sequence of requests and the actual one. It is important to notice that when predicting requests, both a position error and an arrival time error may occur for each request. Then, they proposed the \delaytrust~algorithm, which combines an online algorithm with a predicted route; that is, allowing the server to initially perform a pure online algorithm and then switch to the predicted route at a predefined time point. The server terminates at the origin when receiving an extra \emph{end signal}. 
Hu et~al.~\cite{Hu2022arxiv} introduced another prediction strategy over the arrival time of the last request, rather than a whole sequence of requests. 
While the whole sequence prediction can achieve a good result if the prediction error is small, it may accumulate significant losses over an input sequence if the prediction is poor. 
On the other hand, 
the new prediction strategy over only the last request actually limits the improvement. 
Note that the prediction models proposed in \cite{bernardini2022a} and \cite{Hu2022arxiv} both utilized the concept of \emph{offline predictions}. Precisely, they predicted the entire set of all requests or decided a predictive strategy before conducting an algorithm. 
Very recently, Shao et al. \cite{shao2023online} presented the concept of \emph{online predictions} 
which 
predicts the next online request whenever a new actual request arrives for the OLTSP. Through online predictions, a learning-augmented algorithm can make an adaptive prediction based on the malicious adversary's feedback. It may work better and compensate some loss if an initial (offline) prediction is terrible, but the complexity of its competitive analysis surely
increases. 
Here, we remark that using online predictions does not require knowing the total number of requests $n$ in advance, which contrasts with most prior studies in the learning-augmented framework.
When $n$ is known beforehand, a naive polynomial-time 5/2-approximation algorithm exists, where the server begins the route only after all requests have been released~\cite{nagamochi1997complexity}.

\medskip

\noindent{\bf Our contribution.}
First, we propose a lower bound for any schedule-based algorithms with online predictions. This lower bound demonstrates that, 
given
a schedule even with 
some
upcoming requests that can be perfectly predicted, no schedule-based algorithms with online predictions can achieve a competitive ratio less than 2.

Second, while most learning-augmented algorithms focused on predicting 
a 
full
solution (e.g., a 
whole
route for the OLTSP),  
we look into the waiting duration of the \smartstart~algorithm that achieves the currently best upper bound for the OLTSP and the OLDARP. We propose a learning-augmented algorithm which incorporates \emph{online predictions} into \smartstart, denoted by \ouralg. 
Precisely,
\ouralg~predicts a binary decision for each schedule, choosing either to start 
earlier
or to wait longer, 
comparing to \smartstart's waiting strategy. It takes advantage of online predictions to reduce the impact of significant prediction errors that may occur with offline predictions.
While the competitive ratio of \smartstart~is $\max \{ \theta, \approxNT (1+\frac{1}{\theta-1} ), \frac{\theta}{2}+\approxNT\}$, where $\theta$ is a given waiting scaling parameter and $\approxNT$ is an approximation ratio of a schedule, \ouralg~achieves the upper bound: 
$\max\{\lambda \theta+\frac{\lambda \theta \varepsilon_f}{\costopt}, \approxNT (1+ \frac{1}{\frac{\theta}{\lambda}-1} ) + \frac{\varepsilon_f}{\costopt}, (\frac{\lambda\theta}{2}+\approxNT) \}$-consistency
and  $\max\{ \frac{\theta}{\lambda},\approxNT (1+ \frac{1}{\lambda\theta-1} ),  (\frac{\theta}{2\lambda}+\approxNT)\}$-robustness, 
where 
$\lambda$ is a confidence level parameter and 
$\lambda \in (\frac{1}{\theta},1]$.
We refer to~\cite{kumar2018improving_ski} and derive a balance between the performance from \smartstart~and online predictions, leading to a trade-off between consistency and robustness. Note that when the value of $\theta$ is set to be $\frac{1+\sqrt{13}}{2} \approx 2.3028$, it not only allows \smartstart~to achieve the ratio of 2.6514 with an approximate schedule using the Christofides' algorithm ($\approxNT = 1.5$), but lets
\ouralg~achieve a competitive ratio of about $(1.1514\lambda+1.5)$-consistency and $(1.5+\frac{1.5}{2.3028\lambda-1})$-robustness 
(see Section~\ref{sec:alg} (Corollary~\ref{cor:result})).
One can observe that \ouralg~can achieve its lowest consistency ratio when $\lambda$ is close to $\frac{1}{\theta}$, which almost meets the lower bound of 2. That is, 
our upper bound approaches the optimal ratio very closely.

\begin{table}[h!]
\centering
\begin{tabular}{|p{1.2cm}|p{4cm}|p{9.7cm}|}
\hline
\textbf{Bound} & \smartstart \cite{ascheuer2000online} & \text{\smartstart~with Online Predictions (\ouralg)} \\
\hline
\textbf{UB} & 
$\max\left\{ \theta,\, \rho\left(1 + \frac{1}{\theta - 1}\right),\, \frac{\theta}{2} + \rho \right\}$ & 
$\min\left\{
\begin{aligned}
& \max\left\{ \lambda \theta + \frac{\lambda \theta \varepsilon_f}{\costopt},\, \approxNT \left(1 + \frac{1}{\frac{\theta}{\lambda} - 1} \right) + \frac{\varepsilon_f}{\costopt},\, \left(\frac{\lambda \theta}{2} + \approxNT\right) \right\}, \\
& \max\left\{ \frac{\theta}{\lambda},\, \approxNT \left(1 + \frac{1}{\lambda \theta - 1} \right),\, \left( \frac{\theta}{2\lambda} + \approxNT \right) \right\}
\end{aligned}
\right\}$ \\
\hline
\textbf{UB} \parbox[t]{2cm}{($\rho = 1.5$)} & 2.6514 & 
$\min\left\{ 1.1514\lambda + 1.5,\, 1.5 + \frac{1.5}{2.3028\lambda - 1} \right\}$ \\
\hline
\textbf{LB} & $2$ \cite{phdthesis2020} & 2 \\
\hline
\end{tabular}
\caption{Comparison between \smartstart~and \ouralg~on UB, UB ($\rho = 1.5$), and LB, where $\lambda \in (\frac{1}{\theta}, 1]$}
\end{table}

\noindent{\bf Technical Overview.}
Here, we present the rationale behind 
our prediction strategy, i.e.,
predicting a binary decision for each schedule. 
As mentioned, the currently best \smartstart~algorithm is a non-zealous schedule-based algorithm for the OLTSP and the OLDARP. 
In a schedule-based approach, once the server 
begins performing
a schedule, it ignores all new incoming requests that arrive during the 
current schedule.
These requests will instead be considered in the construction of the next schedule. Furthermore, being non-zealous means the server does not immediately begin a schedule as soon as unserved requests 
exist or arrive.
Instead, 
the server
waits at the origin to determine a proper starting time. 
The key idea behind \smartstart~is to strategically delay the start of each schedule in order to potentially include 
sufficient incoming
requests. A longer waiting period may allow the server to collect more requests, 
leading
to 
a more effective schedule.
However, if no new requests arrive during this 
period,
the delay results in wasted waiting. In such 
a scenario,
starting the schedule earlier could reduce the overall completion time.

To address this trade-off, we propose using 
online
predictions to make a binary decision for each schedule, choosing either to start 
earlier
or to wait longer. 
This approach aims to wait \emph{intelligently} at the origin by relying on online predictions that adapt to the dynamics of incoming requests.

\section{Preliminary}
For the 
online traveling salesman problem (OLTSP),
we consider a sequence of requests that arrive over time, denoted by 
$X = \{x_1,x_2,...\}$. Importantly, the sequence is revealed in an online fashion, and the total number of requests $n$ is not known in advance.
Each request $x_i$ is represented by $x_i = (a_i, t_i)$, where request $x_i$ 
arrives 
at
position $a_i$ at time $t_i$.
Note that the server can 
serve $x_i$ 
only
after time $t_i$. 
Regarding 
the extension of
the OLTSP, 
the online dial-a-ride problem (OLDARP), each request $x_i$, arriving 
at
position $a_i$ at time $t_i$, needs to be delivered to $b_i$, denoted by $x_i = (a_i, b_i, t_i)$, in comparison with the OLTSP in which each request has a single position to be served.


In both the OLTSP and the OLDARP, we assume that the server moves with unit speed. The objective is to find a route, denoted by $R_X$, which allows the server to serve all requests in $X$ and returns to the origin such that the completion time is minimized. 
We also let $d(p, p^\prime)$ represent the shortest distance between positions $p$ and $p^\prime$ in the metric space.

\medskip

\noindent{\bf Competitive analysis.}
Let the cost of an online algorithm, {ALG}, be denoted by $|\text{ALG}|$. We evaluate the performance of an online algorithm by comparing it with an offline optimum, denoted by $\costopt$, which is derived by the malicious adversary who knows all future information in advance. 
An online algorithm is $\alpha$-competitive if 
it satisfies $|\text{ALG}| \leq \alpha \costopt $ for any input instance, where 
$\alpha$ is called a competitive ratio of {ALG}. 

\medskip

\noindent{\bf Consistency and robustness.}
We refer to~\cite{kumar2018improving_ski} about the definition of the performance of a learning-augmented algorithm. 
One can incorporate a prediction cost function, in terms of 
a prediction error $\eta$, denoted by $c(\eta)$, and 
present the competitive ratio of a learning-augmented algorithm using the function $c(\eta)$. That is, 
an algorithm is $\gamma$-robust if $c(\eta)\leq \gamma$ for all the value of $\eta$. On the other hand, an algorithm is $\beta$-consistent if $c(0) = \beta$.

\medskip

\noindent{\bf Schedule-based algorithms.}
The concept of schedule-based algorithms was discussed by Brix~\cite{phdthesis2020} and Yu et al.~\cite{zhang}. 
A schedule-based algorithm operates by waiting at the origin for a certain amount of time and then computing a schedule, $R_{U_t}$, to serve the set of currently known and unserved requests $U_t$ at time $t$, treating 
the set
as an offline instance. Once the schedule $R_{U_t}$ begins, all new incoming requests are ignored until it is completed. In other words, 
each
schedule 
handles
only the requests available at its starting time.
Since we focus on schedule-based algorithms, the 
target
route $R_X$ 
actually comprises
a sequence of schedules, where each schedule serves a batch of available requests without interruption, and any new requests that arrive during its execution are deferred to future schedules.

Note that \ignore~\cite{ascheuer2000online} and \smartstart~\cite{ascheuer2000online} are schedule-based online algorithms for the OLTSP and 
the OLDARP.
That is, when the server is currently 
performing a schedule, if any new requests arrive, the server ignores them and continues its current schedule. The difference between them is that \ignore~is a zealous \cite{zealous2001} schedule-based algorithm, implying that the server performs a schedule immediately from the origin as long as any unserved requests are released. On the other hand, \smartstart~can wait for a moment at the origin and begin to serve unserved requests later; i.e., \smartstart~is non-zealous. Obviously, based on the merit of the waiting gadget, \smartstart~can achieve a better upper bound on competitive ratios.

The notion of \emph{online predictions} was first introduced by Shao et al.~\cite{shao2023online}. They defined that a prediction of the next online request can be made only upon the arrival of the current request. Specifically, they predicted the arrival time and position of the next online request whenever a new request arrives. 
In their setting, predictions are made incrementally over individual requests, rather than all at once.
In this study, we adapt the idea of online predictions from the request level to the schedule level. 
Instead of making a prediction for each incoming request, we allow a single prediction to be made only when the server returns to the origin and a schedule has just finished. 


\begin{definition}(Schedule-based algorithm with online predictions)\label{def:schedule_based_op}
For any schedule-based algorithm ALG with online predictions, it operates at time~$t$ as follows:  
    {ALG} makes a prediction about how to perform the next schedule only if the server is idle or waiting at the origin.
\end{definition}

Most prior studies usually predicted the entire set of all online requests or decided a predictive strategy in the beginning, and then conducted their algorithms in an offline manner based on the predictions. Thus we called the type of predictions \emph{offline predictions}. Apparently offline predictions can help achieve a better upper bound than online predictions if the predictions are accurate. However, large prediction errors, if any, may cause significant trouble when performing an online algorithm. In contrast, online predictions can provide a predictor with more flexibility to adaptively adjust its accuracy over time, which might be helpful in complicated scenarios.

\section{Lower Bound for Schedule-based Algorithms with Online Predictions}\label{sec:LB}

According to Definition~\ref{def:schedule_based_op}, we present a lower bound for any schedule-based algorithms with online predictions for 
the OLTSP. Since the OLDARP is a generalization of the OLTSP, the lower bound for the OLTSP is applicable to the OLDARP. 
In particular, 
we show the lower bound for the OLTSP on the real line, the one-dimensional case in the metric space, and it improves the previous result in Shao et al.~\cite{shao2023online} which proved a lower bound of 1.75 for any zealous algorithms with perfect online predictions.

\begin{theorem}\label{thm:LB}
    Let $A$ be a schedule-based algorithm with online predictions for the OLTSP on the line. The competitive ratio of $A$ is at least 2.
\end{theorem}

\begin{proof}
To show the lower bound, by Definition~\ref{def:schedule_based_op}, when the server is at the origin and the current schedule is finished, $A$ can make a prediction over future requests or their information to build the next schedule. Then it can decide when the server performs the next schedule.
Note that $A$ is a schedule-based algorithm with online predictions. Thus, if there is a new incoming request that is not in the schedule, the server will ignore it.  

Initially, assume a request $x_0 = (1, 0)$ 
arrives at position 1 at time 0.
Then, there are some incoming requests, where $x_1=(1-\varepsilon,1+\varepsilon)$, $x_2=(1-2\varepsilon,1+2\varepsilon), \ldots$, $x_k=(1-k\varepsilon,1+k\varepsilon)$, satisfying $0<\varepsilon=\frac{1}{k}$, for any integer $k$.
At time 0, $A$ predicts $k$ requests, from $\hat{x}_1$ to $\hat{x}_k$, and let a schedule comprise $x_0$ and $\hat{x}_1$ to $\hat{x}_k$. 
Suppose $x_1, \ldots$, $x_k$ are all correctly predicted by $A$, i.e., 
$ \hat{x}_i = x_i=(a_i,t_i)$,
$1 \le i \le k$. We conduct the analysis depending on the start time of the first schedule by $A$, denoted by $t_1$. 

We examine each case depending on the position of the server at time $t_1$. Note that the blue line represents the walk of OPT, while the orange line represents the walk of $A$. We denote the initial request $x_0=(1,0)$ as a black point, actual requests $x_1, \ldots$, $x_k$ as dark green points and predicted requests $\hat{x}_1, \ldots$, $\hat{x}_k$ as light green points.

\begin{itemize}
    \item {Case 1: $0\leq t_1<1$.}\
    In this scenario, the server performs the first schedule before time~1. 
    Then, an actual request $x_{k+1}=(1,1)$ arrives at time 1. According to Definition~\ref{def:schedule_based_op}, the server ignores $x_{k+1}$ in the first schedule. That is, the second schedule needs to be built to serve $x_{k+1}$.
    The completion time of $A$ is $|A|\geq t_1+2+2\geq 4$.
    However, the optimal offline adversary i.e., OPT, can serve all of these requests when they arrive, and complete them at time $2$. 
    Thus, we can obtain $|A| \geq 2 \costopt$.
        \begin{figure}[htb] 
            \centering 
            \includegraphics[width=0.5\textwidth]{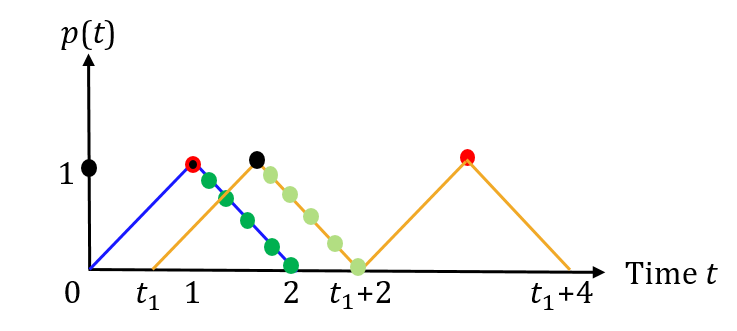}
            \caption{Case 1: $0\leq t_1<1$ and the red point represents $x_{k+1}=(1,1)$}
            \label{LB_case1}
        \end{figure}
        
    \item {Case 2: $1\leq t_1< 2$.}
    For $0\leq h < k$, $h \in \mathcal{N}$, we assume $1+h\varepsilon\leq t_1<1+h\varepsilon+ \frac{\varepsilon}{2}$. An actual request $x_{k+1}=(1-h\varepsilon-\frac{\varepsilon}{2},1+h\varepsilon+\frac{\varepsilon}{2})$ arrives. Note that $h\varepsilon < k\varepsilon =1$ since $h<k$.
    Similar to Case 1, the server ignores $x_{k+1}$ in the first schedule and the second schedule needs to be constructed to serve $x_{k+1}$. The duration of the second schedule is at least $2(1-h\varepsilon-\frac{\varepsilon}{2})$. As $t_1\geq 1+ h \varepsilon$, the completion time of $A$ is $|A|\geq t_1+2+2(1-h\varepsilon-\frac{\varepsilon}{2})\geq 5-h\varepsilon-\varepsilon$. Thus, we obtain $|A| > 4- \varepsilon$; meanwhile, OPT is still 2. Finally, we also conclude $|A| \geq 2 \costopt$ when $k \to \infty$. 
        
        \begin{figure}[htb] 
            \centering 
            \includegraphics[width=0.55\textwidth]{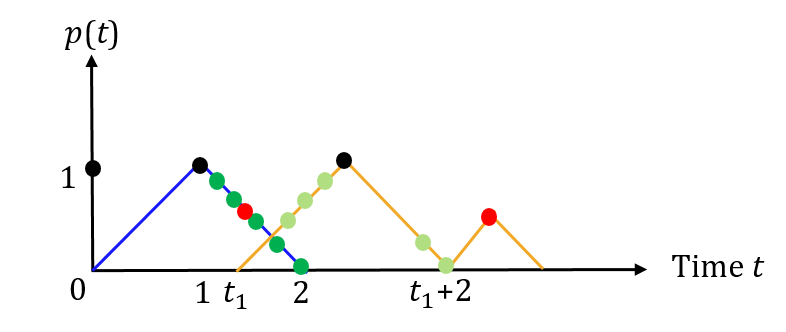}
            \caption{Case 2: $1\leq t_1< 2$ and the red point represents $x_{k+1}=(1-h\varepsilon-\frac{\varepsilon}{2},1+h\varepsilon+\frac{\varepsilon}{2})$}
            \label{LB_case2}
        \end{figure}
    \item {Case 3: $ t_1\geq 2$.}
    In this case, let an actual request $x_{k+1}$ be $(1-\delta,1+\delta)$ where $ 0\leq \delta \leq 1$, and $x_{k+1} \notin \{ x_1,x_2,...,x_k\}$.
    The server can serve all the requests in the first schedule if $t_1\geq 2$. Though, OPT can serve all the requests at time 2, while the completion time of $A$ is $|A| \geq t_1 + 2 \geq 4$. Thus, we also can obtain $|A| \geq 2 \costopt$. 
        \begin{figure}[htb] 
            \centering 
            \includegraphics[width=0.55\textwidth]{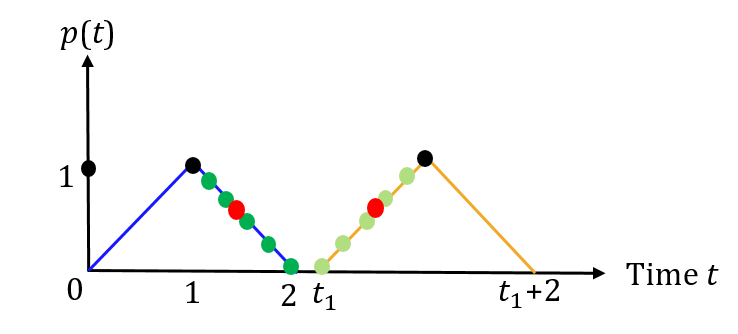}
            \caption{Case 3: $ t_1\geq 2$ and the red point represents $x_{k+1} = (1-\delta,1+\delta)$}
            \label{LB_case3}
        \end{figure}
\end{itemize}
\end{proof}

\section{Algorithm}\label{sec:alg}
We first recall the \smartstart~algorithm, which was proposed 
by Ascheuer et al.~\cite{ascheuer2000online}, reproduced here as Algorithm 1 for self-containedness.

\begin{table}[htb]\caption{Algorithm 1 \smartstart~\cite{ascheuer2000online}}
\label{alg:smartstart}
\begin{center}
\begin{tabular}{l}
\hline
\noalign{\smallskip}
    Given a waiting scaling parameter $\theta$ and an approximation ratio of a schedule \\ denoted by $\approxNT$, 
    at any time $t$, the server of \smartstart~enters one of the following states:\\ 
\noalign{\smallskip}
\hline
\noalign{\smallskip}
\begin{minipage}{4.5in}
    \vskip 2pt
    \begin{enumerate}
      \item[1.] Idle: The server is idle at the origin as there are no unserved requests.\\ 
      \item[2.] Waiting: The server is at the origin. When there exists a non-empty set of unserved requests $U_t$, the server is not idle. Let $R_{U_t}$ denote the schedule constructed from $U_t$, and let $|R_{U_t}|$ denote its length. If the server cannot finish $R_{U_t}$ by time $\theta t$, that is, if $t + |R_{U_t}| > \theta t$, or equivalently, $t < \frac{|R_{U_t}|}{\theta - 1}$, the server remains waiting at the origin. Otherwise, if $t + |R_{U_t}| \leq \theta t$ (i.e., $t \geq \frac{|R_{U_t}|}{\theta - 1}$), the server starts performing the schedule $R_{U_t}$.\\
      \item [3.] Working: 
      The server is currently performing the schedule $R_{U_{t_p}}$, where $R_{U_{t_p}}$ denotes the preceding schedule, as it was constructed before the current time. That is, $t_p + |R_{U_{t_p}}| \leq \theta t_p$.
      If any new requests arrive during the period of $R_{U_{t_p}}$, i.e., $[t_p,, t_p + |R_{U_{t_p}}|]$, the server ignores them and continues with the current schedule.\\
    \end{enumerate}
    \vskip 2pt
\end{minipage}
\vspace{-5pt}
\\
\hline
\end{tabular}
\end{center}
\end{table}

The server has three states: (1) The \emph{working} state indicates that the server is currently performing a schedule. In this state, if any new requests appear, 
similar to the behavior of \ignore~described in~\cite{ascheuer2000online}, 
the server ignores them and continues with the current schedule. (2) The \emph{idle} state indicates that the server is idle at the origin and there are no unserved requests at that time. (3) The \emph{waiting} state indicates that 
at the current time $t$, 
the server is waiting at the origin 
until a \emph{signal} to start processing 
unserved requests that just arrive, denoted by a set $U_t$, at time $t$. 
To determine the 
signal to initiate work,~\smartstart~incorporates a waiting parameter $\theta$.
Specifically,
when the server is at the origin and there exists a set of unserved requests $U_t$ at time $t$,~\smartstart~determines the state of the server. 
If the next schedule $R_{U_t}$ formed by the unserved requests in $U_t$ can be finished by time $\theta t$ at the latest, the server enters the working state. Otherwise, it 
remains waiting. 


\begin{theorem}[\cite{ascheuer2000online}, Theorem 6.]\label{thm:smartstart}
Given a waiting scaling parameter $\theta$ and an approximation ratio of a schedule denoted by $\approxNT$,
it holds that for any $\theta \geq \approxNT$ ($\theta > 1$ for sure),~\smartstart~is $\alpha$-competitive with 
    \begin{center}
         $ \alpha = \max \{ \theta, \approxNT (1+\frac{1}{\theta-1} ), \frac{\theta}{2}+\approxNT\}$. 
    \end{center}
\end{theorem}

\medskip

~\smartstart~achieves the best competitive ratio of $\frac{1}{4}(4\approxNT+1+\sqrt{1+8\approxNT})$ when the choice of $\theta$ is 
set to be $\frac{1}{2}(1+\sqrt{1+8\approxNT})$~\cite{ascheuer2000online}. 
Note that for the Online Traveling Salesman Problem (OLTSP), where Christofides' algorithm yields a polynomial time approximation algorithm with $\approxNT = 1.5$,~\smartstart~achieves 
the currently best competitive ratio of $\frac{7+\sqrt{13}}{4} \approx 2.6514$ 
in polynomial time, 
when $\theta = \frac{1+\sqrt{13}}{2} \approx 2.3028$.

\longdelete{
\begin{algorithm}[htb]
    \begin{algorithmic}[1] 
    \caption{\textsc{Smartstart~\cite{ascheuer2000online}}}\label{alg:smartstart}
    \Require A set of unserved requests $U_t$ at time ${t}$,  a given $\theta$ with $\theta \geq \approxNT$
    \While {$U_t \neq \emptyset$}
        \If {the server is at the origin $o$}
            \State compute an approximate schedule $R_{U_t}$;
        \EndIf
        \While {$t<\frac{|R_{U_t}|}{\theta-1}$}
            \State wait at the origin $o$;
        \EndWhile
        \State $R_{U_t}\gets$ the schedule formed by the unserved requests; perform $R_{U_t}$;
    \EndWhile
    \end{algorithmic}
\end{algorithm}
}



In the following, 
we present a \emph{learning-augmented algorithm} 
which incorporates 
online predictions into~\smartstart,
making a simple binary decision for each schedule,
denoted by~\ouralg.
We introduce 
a new
\emph{predicting} state into \smartstart~when the server is located at the origin, considering a set of unserved requests $U_t$. 
While a natural approach might be 
predicting
the requests in $U_t$,
our method does not require forecasting specific requests. 
Instead, we 
simply make
a binary decision for performing the next schedule.
Precisely,
an approximate schedule denoted by $R_{U_t}$, which comprises the unserved requests in $U_t$, is first built, and the start time of $R_{U_t}$ has to be determined by~\smartstart, denoted by $t^\prime$. 
Simultaneously, a binary prediction is made to decide how to perform the next schedule. Specifically, the predictor determines whether the server should wait longer at the origin, starting later than $t^\prime$, in anticipation of future requests, or start 
early
with the currently known requests and perform the schedule earlier than $t^\prime$. Based on this 
prediction,
one of the two decisions is selected, corresponding to either the late-start gadget or the early-start gadget.
If \ouralg~selects the late-start gadget, it determines whether the completion time of the 
schedule $R_{U_t}$, i.e., $t+{|R_{U_t}|}$, is greater than $\lambda\theta t$ or not, where $\lambda$ is a parameter satisfying $\lambda \in (\frac{1}{\theta},1]$. 
Once the 
schedule $R_{U_t}$ can be finished by time $\lambda \theta t$ at the latest, the server enters the working state. Otherwise, it waits at the origin. On the other hand, if \ouralg~selects the early-start gadget, it examines the relationship between $t+{|R_{U_t}|}$ and  $\frac{\theta}{\lambda} t$. 
See Algorithm 2 for the details.

\begin{theorem}\label{thm:ouralg}
Given a waiting scaling parameter $\theta$ and an approximation ratio $\approxNT$ of a schedule,  
the competitive ratio of~\ouralg~is 
$\max\{\lambda \theta+\frac{\lambda \theta \varepsilon_f}{\costopt},\approxNT (1+ \frac{1}{\frac{\theta}{\lambda}-1} ) + \frac{\varepsilon_f}{\costopt}, (\frac{\lambda\theta}{2}+\approxNT) \}$-consistent
and $\max\{ \frac{\theta}{\lambda},\approxNT (1+ \frac{1}{\lambda\theta-1} ),  (\frac{\theta}{2\lambda}+\approxNT)\}$-robust, 
where $\lambda \in (\frac{1}{\theta},1]$ 
and $\varepsilon_f = | \hat{t}_f - t_f|$.
\end{theorem}

To simplify the description, 
we use new notation, 
$\max \{ A, B, C\}$,
to represent the competitive ratio $\alpha$ of \smartstart. 
Precisely, 
as shown in Table~\ref{tab:ratio_smartstart}, $A$ represents $\theta$, $B$ represents $ \approxNT (1+\frac{1}{\theta-1} )$, and $C$ represents $\frac{\theta}{2}+\approxNT$.

\begin{table}[htb]
\centering
\caption{New notation 
for simplifying each item in 
the ratio of \smartstart}\label{tab:ratio_smartstart}
\begin{tabular}{|c|c|c|c|}
\hline 
Items 
& $\theta$ & $ \approxNT (1+\frac{1}{\theta-1} )$ & $\frac{\theta}{2}+\approxNT$ \\
\hline 
New notation & $A$ & $B$ & $C$ \\
\hline 
\end{tabular}
\end{table}

We also simplify the result of~\ouralg, as shown in Table~\ref{tab:ratio_alg}.
We denote the consistency 
as $\max\{ \lambda A +\frac{\lambda \theta \varepsilon_f}{\costopt}, B_{\lambda} + \frac{\varepsilon_f}{\costopt}, C_\lambda \}$,
where 
$B_{\lambda}$ represents $\approxNT (1+ \frac{1}{\frac{\theta}{\lambda}-1} )$ and $C_\lambda$ represents $(\frac{\lambda\theta}{2}+\approxNT)$.
Similarly, 
we denote the robustness 
as $\max = \{ \frac{A}{\lambda}, B_{1/\lambda}, C_{1/\lambda} \}$, 
where 
$B_{1/\lambda}$ represents $\approxNT (1+ \frac{1}{\lambda\theta-1} )$ and $C_{1/\lambda}$ represents $(\frac{\theta}{2\lambda}+\approxNT)$.
The simplification 
helps clarify the trade-off between consistency and robustness.
\begin{table}[hp] 
\centering
\caption{New notation 
for
simplifying 
each item
in the ratios of~\ouralg~with 
$\lambda \in (\frac{1}{\theta},1]$}\label{tab:ratio_alg}
\begin{tabular}{|c|c|c|c|c|c|c|c|c|c|c|c|c|c|}
\hline   
\multirow{2}{*}{Consistency} & Items & \multicolumn{4}{c|}{$\lambda \theta$} &  \multicolumn{4}{c|}{$\approxNT (1+ \frac{1}{\frac{\theta}{\lambda}-1} )$} & \multicolumn{4}{c|}{$\frac{\lambda\theta}{2}+\approxNT$} \\\cline{2-14}
                       & New notation & \multicolumn{4}{c|}{ $\lambda A$} & \multicolumn{4}{c|}{ $B_{\lambda}$ } & \multicolumn{4}{c|}{ $C_\lambda$ } \\\cline{1-1}\cline{2-14}

\hline
\hline
                       
\multirow{2}{*}{Robustness} & Items & \multicolumn{4}{c|}{$\frac{\theta}{\lambda}$} & \multicolumn{4}{c|}{$ \approxNT (1+ \frac{1}{\lambda\theta-1} )$}  & \multicolumn{4}{c|}{$\frac{\theta}{2\lambda}+\approxNT$} \\\cline{2-14}
                       & New notation & \multicolumn{4}{c|}{$\frac{A}{\lambda}$} & \multicolumn{4}{c|}{$B_{1/\lambda}$} & \multicolumn{4}{c|}{ $C_{1/\lambda}$ } \\\cline{1-1}\cline{2-14}
\end{tabular}
\end{table}

\begin{corollary}\label{cor:result}
When applying Christofides' $\approxNT$-approximation algorithm to 
a schedule with $\approxNT = 1.5$, \ouralg~achieves its best possible ratio, i.e., $(1.1514\lambda+1.5)$-consistency and $(1.5+\frac{1.5}{2.3028\lambda-1})$-robustness, where 
$\lambda \in (\frac{1}{\theta},1]$
and  $\theta$ is set to be $\frac{1+\sqrt{13}}{2} \approx 2.3028$.
\end{corollary}

\begin{table}[hp]\caption{Algorithm 2 \smartstart~with online predictions (\ouralg)}
\label{alg:ouralg}
\begin{center}
\begin{tabular}{l}
\hline
\noalign{\smallskip}
    Given a waiting scaling parameter $\theta$, an approximation ratio of a schedule, \\ denoted by $\approxNT$, and the confidence level $\lambda \in (\frac{1}{\theta},1]$, 
    at any time $t$, the server \\ of \ouralg~enters one of the following states:\\ 
\noalign{\smallskip}
\hline
\noalign{\smallskip}
\begin{minipage}{4in}
    \vskip 2pt
    \begin{enumerate}
      \item[1.] Idle: The server is idle at the origin as there are no unserved requests.\\ 
      \item[2.] Predicting: The server is at the origin. When there exists a non-empty set of unserved requests $U_t$, the server is not idle. Let $R_{U_t}$ denote the schedule constructed from $U_t$, and let $t^\prime$ denote the start time of $R_{U_t}$ as determined by~\smartstart. Simultaneously, a binary prediction is made to decide whether the server should wait longer at the origin and begin after $t^\prime$, in anticipation of additional future requests, or start earlier than $t^\prime$ with the currently known requests. Accordingly, if the predictor recommends starting later than $t^\prime$, \ouralg~selects the 1st gadget (late-start); otherwise, if starting earlier is advised, \ouralg~selects the 2nd gadget (early-start).\\
      \item [3.] Waiting: According to the prediction state, \ouralg~enters one of the following gadgets accordingly:\\
          \begin{enumerate}
                \item [(i)] 1st gadget (late-start): 
                If the server cannot finish $R_{U_t}$ by time $\lambda \theta t $, that is, if $t + |R_{U_t}|>\lambda \theta t $, or equivalently, $t < \frac{|R_{U_t}|}{\lambda \theta-1}$, the server remains waiting at the origin. Otherwise, if $t + |R_{U_t}| \leq \lambda \theta t $ (i.e., $t \geq \frac{|R_{U_t}|}{\lambda \theta-1}$), the server starts performing the schedule $R_{U_t}$.\\
                \item [(ii)] 2nd gadget (early-start): 
                If the server cannot finish $R_{U_t}$ by time $\frac{\theta}{\lambda}t $, that is, if $t + |R_{U_t}|>\frac{\theta}{\lambda} t $, or equivalently, $t < \frac{|R_{U_t}|}{\frac{\theta}{\lambda}-1}$, the server remains waiting at the origin. Otherwise, if $t + |R_{U_t}| \leq \frac{\theta}{\lambda} t $ (i.e., $t \geq \frac{|R_{U_t}|}{\frac{\theta}{\lambda}-1}$), the server starts performing the schedule $R_{U_t}$.\\
            \end{enumerate}
      \item [4. ] Working: 
      The server is currently performing the schedule $R_{U_{t_p}}$, where $R_{U_{t_p}}$ denotes the preceding schedule, as it was constructed before the current time. 
      That is, it satisfies either $t_p + |R_{U_{t_p}}| \leq \lambda \theta t_p$ or $t_p + |R_{U_{t_p}}| \leq \frac{\theta}{\lambda}t_p$, depending on the gadget under which the schedule was initiated.
      If any new requests arrive during the interval $[ t_p , t_p + |R_{U_{t_p}}| ]$, the server ignores them and continues with the current schedule.\\
    \end{enumerate}
    \vskip 2pt
\end{minipage}
\vspace{-5pt}
\\
\hline
\end{tabular}
\end{center}
\end{table}

\section{Analysis}

Let $t_n$ 
denote
the arrival time of the 
last
request, i.e., $x_n = (a_n,b_n,t_n)$.
We then divide the analysis into two parts, as shown in Lemma~\ref{lem:waiting} and Lemma~\ref{lem:working}.
In the former, 
we focus on the scenario where the server is not engaged in a schedule at time $t_n$.
In the latter, we consider the other scenario where the server is actively performing a schedule at time $t_n$, indicating that it is in the working state.
Note that $t_i \leq |\text{OPT}|$ for any $i$ and $|R_{U_{t}}| \leq \approxNT |\text{OPT}|$ for any~$t$. 
Table~\ref{tab:notation} summarizes the notation that are used in the rest of the paper.


\begin{table}[h]
\centering
\caption{Notation}\label{tab:notation}
\begin{tabular}{ll}
\hline
Notation & Definition\\
\hline
$|\textsc{SSOP}|$ & The completion time (i.e. cost) of~\ouralg\\
$\costopt$ & The completion time (i.e. cost) of~\text{OPT}\\
$R_{U_{t_n}}$ & The schedule that serves the last request $x_n$\\
$U_{t_n}$ & The set of requests scheduled together with $x_n$ in $R_{U_{t_n}}$\\
$t^\prime_f$ & The time used by~\smartstart~to determine when to initiate $R_{U_{t_n}}$\\
$\hat{t}_f$ & The time used by~\ouralg~to determine when to initiate $R_{U_{t_n}}$\\
$R_{U_{\hat{t}_p}}$ & The schedule that immediately precedes the last schedule $R_{U_{t_n}}$\\
$U_{\hat{t}_p}$ & The set of requests scheduled in $R_{U_{\hat{t}_p}}$\\
$\hat{t}_p$  & The time used by~\ouralg~to determine when to initiate $R_{U_{\hat{t}_p}}$\\
$\varepsilon_f$ & The time difference between $\hat{t}_f$ and $t_n$, i.e.,  $\varepsilon_f=| \hat{t}_f - t_n|$ \\
\hline
\end{tabular}
\end{table}


\begin{table}[h]
\centering
\caption{The six cases in Lemma \ref{lem:waiting} involving the scenario in which the server is not working at the origin at time $t_n$}\label{tab:lem_waiting}
\begin{tabular}{|c|c|c|c|}
\hline
\ouralg's gadget  & Case  & Scenario & Competitive ratio\\
\hline
\multirow{3}{*}{Late-start} & Case 1  & $t^\prime_f \leq t_n \leq \hat{t}_f$ & $\min\{ \lambda A + \frac{\lambda \theta \varepsilon_f}{\costopt}, B_{1/\lambda} \}$\\\cline{2-4}
    & Case 2 & $t_n < t^\prime_f \leq \hat{t}_f$ & 
    $\min\{ \lambda A + \frac{\lambda \theta \varepsilon_f}{\costopt}, B_{1/\lambda} \}$\\\cline{2-4}
    & Case 3 & $t^\prime_f \leq \hat{t}_f \leq t_n $ & $\lambda A$ \\\cline{1-4}
\multirow{3}{*}{Early-start} & Case 4 & $\hat{t}_f \leq t_n \leq t^\prime_f $ & $min\{ B_\lambda + \frac{\varepsilon_f}{\costopt}, \frac{A}{\lambda}\}$ 
\\\cline{2-4}
    & Case 5 & $t_n < \hat{t}_f \leq t^\prime_f$ &  $B_\lambda$\\\cline{2-4}
    & Case 6 & $\hat{t}_f \leq t^\prime_f \leq t_n $ & 
    $min\{ B_\lambda + \frac{\varepsilon_f}{\costopt}, \frac{A}{\lambda}\}$
    \\\cline{1-4}
\end{tabular}
\end{table}

We first consider a situation in which the server is not working at time $t_n$, i.e., waiting at the origin. At the same time, the last request $x_n$ arrives, indicating the schedule $R_{U_{t_n}}$ is built.
As a result, the server enters the predicting state, and 
we analyze the relationship between 
three key time points
in Lemma~\ref{lem:waiting}: the current time $t_n$, 
the decision time $t^\prime_f$ at which~\smartstart~determines when to initiate $R_{U_{t_n}}$
and the decision time $\hat{t}_f$ at which~\ouralg~determines the same schedule.

\begin{lemma}\label{lem:waiting}
Suppose the server is not working at time $t_n$.
The competitive ratio of~\ouralg~is 
$\max\{\lambda \theta+\frac{\lambda \theta \varepsilon_f}{\costopt}, \approxNT (1+ \frac{1}{\frac{\theta}{\lambda}-1} ) + \frac{\varepsilon_f}{\costopt} \}$-consistent and $\max\{ \approxNT (1+ \frac{1}{\lambda \theta-1}), \frac{\theta}{\lambda}\}$-robust, where $\lambda \in (\frac{1}{\theta},1]$ 
and $\varepsilon_f = | \hat{t}_f - t_n|$.
\end{lemma}

\begin{proof}
The cases discussed in Lemma \ref{lem:waiting} are summarized in Table \ref{tab:lem_waiting}. Cases 1 through 3 involve \ouralg~performing the schedule $R_{U_{t_n}}$ using the late-start gadget, while Cases 4 through 6 consider the scenarios related to the early-start gadget.
The following figures illustrate each case. 
Here, a solid line segment represents the server being in a working state, performing a schedule that starts from the origin and returns to the origin.
On the other hand, a dotted line segment indicates that the server is waiting, idle or predicting at the origin.

    \begin{itemize}
        \item{Case 1: $t^\prime_f \leq t_n \leq \hat{t}_f$, as shown in Fig.~\ref{new_case1-1}.}\\
        In this case, 
       \ouralg~performs the schedule $R_{U_{t_n}}$ at time $\hat{t}_f$. As $t^\prime_f \leq \hat{t}_f$ 
       (i.e. the late-start case),
       it satisfies $\hat{t}_f = \frac{|R_{U_{t_n}}|}{\lambda \theta-1}$,
        which implies that the server waits at the origin until $\hat{t}_f$ before 
        performing
        the schedule $R_{U_{t_n}}$.
        However,~\smartstart~performs $R_{U_{t_n}}$ immediately at time $t_n$ since $t^\prime_f \leq t_n$.
            \begin{figure}[htb] 
                \centering 
                \includegraphics[width=0.55\textwidth]{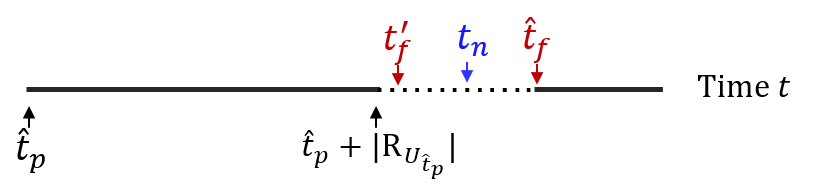}
                \caption{Case 1: $t^\prime_f \leq t_n \leq \hat{t}_f$}
                \label{new_case1-1}
            \end{figure}
            
         The completion time of~\ouralg~is given by $\hat{t}_f + |R_{U_{t_n}}| $, 
         where $|R_{U_{t_n}}|$ represents the total time required to perform $R_{U_{t_n}}$.
         If $\hat{t}_f + |R_{U_{t_n}}| $ is replaced with $ \lambda \theta \hat{t}_f $,
         the first bound is then expressed as follows:
                     \begin{align*}
                \costalg & = \hat{t}_f + |R_{U_{t_n}}|
                = \lambda \theta \hat{t}_f 
                = \lambda \theta (t_n+ \varepsilon_f)
                \\& = \lambda \theta t_n + \lambda \theta \varepsilon_f 
                \leq \lambda \theta \costopt + \lambda \theta \varepsilon_f
            \end{align*}
        For the second bound without the waiting error $\varepsilon_f$, we have
            \begin{align*}
                \costalg & = \hat{t}_f + |R_{U_{t_n}}|
                = \frac{|R_{U_{t_n}}|}{\lambda \theta-1} + |R_{U_{t_n}}|
                \\& = (1+ \frac{1}{\lambda \theta-1})|R_{U_{t_n}}|
                \\&\leq \approxNT (1+ \frac{1}{\lambda \theta-1})\costopt.
            \end{align*}  
        \item{Case 2: $t_n < t^\prime_f \leq \hat{t}_f$, as shown in Fig.~\ref{case1-2}.}\\
        In this scenario, both the server of \ouralg~and \smartstart~wait at the origin at time $t_n$ because $t_n < t^\prime_f$ and $t_n < \hat{t}_f$. 
        However,
        \ouralg~starts performing $R_{U_{t_n}}$ later than \smartstart~due to $t^\prime_f \leq \hat{t}_f$.
            \begin{figure}[htb] 
                \centering  
                \includegraphics[width=0.55\textwidth]{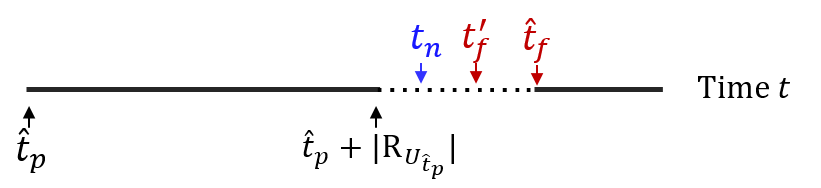} 
                \caption{Case 2: $t_n < t^\prime_f \leq \hat{t}_f$}
                \label{case1-2}
            \end{figure}
        
        Similar to Case 1, the completion time is 
        $\hat{t}_f + |R_{U_{t_n}}|$
        and $\hat{t}_f = t_n + \varepsilon_f$. Therefore, we 
        obtain the first bound:
            \begin{align*}
                \costalg & = \hat{t}_f + |R_{U_{t_n}}|
                = \lambda \theta \hat{t}_f 
                = \lambda \theta (t_n+ \varepsilon_f)
                \\& = \lambda \theta t_n + \lambda \theta \varepsilon_f 
                \leq \lambda \theta \costopt + \lambda \theta \varepsilon_f
            \end{align*}               
        The second bound is also identical to Case 1:
            \begin{align*}
                \costalg & = \hat{t}_f + |R_{U_{t_n}}|
                \leq \approxNT (1+ \frac{1}{\lambda \theta-1})\costopt
            \end{align*}

        \item{Case 3: $t^\prime_f \leq \hat{t}_f \leq t_n $, as shown in Fig.~\ref{case1-3}.}\\ 
        In this case,
        both the server of \ouralg~and \smartstart~begin performing the schedule $R_{U_{t_n}}$ at time $t_n$, 
        immediately upon the arrival of the last request~$x_n$.
            \begin{figure}[htb] 
                \centering 
                \includegraphics[width=0.55\textwidth]{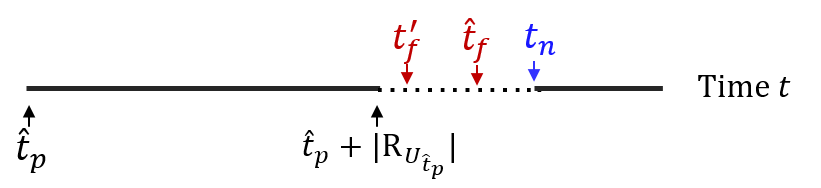} 
                \caption{Case 3: $t^\prime_f \leq \hat{t}_f \leq t_n $}
                \label{case1-3}
            \end{figure}
        
        Since $R_{U_{t_n}}$ is performed at $t_n$
        and $t^\prime_f \leq \hat{t}_f$ 
        (i.e. the late-start case),
        it satisfies $\hat{t}_f= \frac{|R_{U_{t_n}}|}{\lambda \theta -1} $ 
        and $\hat{t}_f \leq t_n$.
        Therefore, the completion time
        of \ouralg~
        is $t_n + |R_{U_{t_n}}|$ no later than $\lambda \theta t_n$. We thus obtain the
        following
        bound:
            \begin{align*}
                \costalg & = t_n + |R_{U_{t_n}}|
                \leq \lambda \theta t_n 
                \leq \lambda \theta \costopt
            \end{align*}
        \item{Case 4: $\hat{t}_f \leq t_n \leq t^\prime_f $, as shown in Fig.~\ref{case1-4}}.\\ 
        Given that $\hat{t}_f \leq t_n$, the server of \ouralg~and \smartstart~starts $R_{U_{t_n}}$ at $t_n$ and $t^\prime_f$, respectively, which demonstrates that \ouralg~achieves an earlier completion of $R_{U_{t_n}}$ compared against \smartstart.
            \begin{figure}[htb] 
                \centering 
                \includegraphics[width=0.55\textwidth]{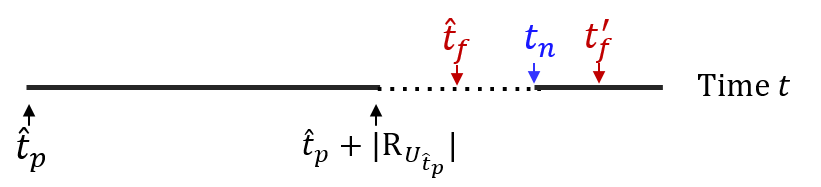} 
                \caption{Case 4: $\hat{t}_f \leq t_n \leq t^\prime_f $}
                \label{case1-4}
            \end{figure}
        
        Specifically,
        the completion time of~\ouralg~is $t_n + |R_{U_{t_n}}|$. In addition, as $\hat{t}_f \leq t^\prime_f$ 
        (i.e. the early-start case),
        it satisfies $\hat{t}_f = \frac{|R_{U_{t_n}}|}{\frac{\theta}{\lambda}-1}$ 
        and $\hat{t}_f \leq t_n$.
        We then obtain the bound:
            \begin{align*}
                \costalg & = t_n + |R_{U_{t_n}}|
                = \hat{t}_f + \varepsilon_f + |R_{U_{t_n}}|
                \\& = \frac{|R_{U_{t_n}}|}{\frac{\theta}{\lambda}-1} +  \varepsilon_f+|R_{U_{t_n}}|
                \\& = (1+ \frac{1}{\frac{\theta}{\lambda}-1} )|R_{U_{t_n}}|+ \varepsilon_f
                \\&\leq \approxNT (1+ \frac{1}{\frac{\theta}{\lambda}-1} )\costopt + \varepsilon_f
            \end{align*}
        Moreover, $t_n \geq  \frac{|R_{U_{t_n}}|}{\frac{\theta}{\lambda}-1}$ 
        since $t_n \geq \hat{t}_f$,
        which
        implies that the completion time $t_n + |R_{U_{t_n}}|$ is no later than $ \frac{\theta}{\lambda} t_n$. Therefore, we 
        obtain the second bound without the waiting error $\varepsilon_f$:
            \begin{align*}
                \costalg & = t_n + |R_{U_{t_n}}|
                \leq \frac{\theta}{\lambda} t_n
                \leq \frac{\theta}{\lambda}\costopt
            \end{align*}
        \item{Case 5: $t_n < \hat{t}_f \leq t^\prime_f$, as shown in Fig.~\ref{case1-5}.}\\
        In this case, 
        both the server of \ouralg~and \smartstart~wait at the origin at $t_n$. However, \ouralg~performs the execution of $R_{U_{t_n}}$ at $\hat{t}_f$, earlier than \smartstart.
            \begin{figure}[htb] 
                \centering 
                \includegraphics[width=0.55\textwidth]{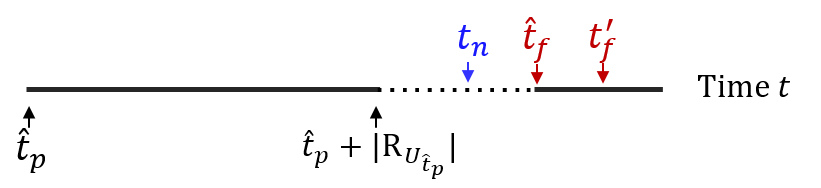} 
                \caption{Case 5: $t_n < \hat{t}_f \leq t^\prime_f$} 
                \label{case1-5}
            \end{figure}
        \\The completion time of~\ouralg~is $\hat{t}_f + |R_{U_{t_n}}|$. In addition, as $ \hat{t}_f \leq t^\prime_f$ (early-start), it satisfies $\hat{t}_f = \frac{|R_{U_{t_n}}|}{\frac{\theta}{\lambda}-1}$. We thus replace $\hat{t}_f$ with $\frac{|R_{U_{t_n}}|}{\frac{\theta}{\lambda}-1}$
        and have the bound: 
            \begin{align*}
                \costalg & = \hat{t}_f + |R_{U_{t_n}}|
                = \frac{|R_{U_{t_n}}|}{\frac{\theta}{\lambda}-1}+ |R_{U_{t_n}}| 
                = (1+ \frac{1}{\frac{\theta}{\lambda}-1} )|R_{U_{t_n}}|
                \\&\leq \approxNT (1+ \frac{1}{\frac{\theta}{\lambda}-1} )\costopt.
            \end{align*}
        \item{Case 6: $\hat{t}_f \leq t^\prime_f \leq t_n $, as shown in Fig.~\ref{case1-6}.}\\
        In this case, both the server of \ouralg~and \smartstart~begin performing $R_{U_{t_n}}$ at time $t_n$, 
        immediately upon the arrival of the last request~$x_n$.
        
            \begin{figure}[htb] 
                \centering 
                \includegraphics[width=0.55\textwidth]{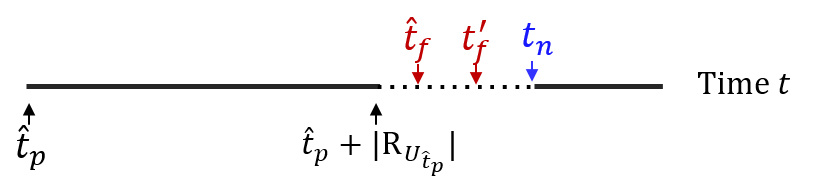}  
                \caption{Case 6: $\hat{t}_f \leq t^\prime_f \leq  t_n $}
                \label{case1-6}
            \end{figure}
        
        Since $R_{U_{t_n}}$ is performed at time $t_n$ and $\hat{t}_f \leq t^\prime_f$ 
        (i.e. the early-start case),
        it satisfies $\hat{t}_f = \frac{|R_{U_{t_n}}|}{\frac{\theta}{\lambda}-1}$. 
        We thus obtain the same bound as Case 4:
            \begin{align*}
                \costalg & = t_n + |R_{U_{t_n}}|
                = \hat{t}_f + \varepsilon_f + |R_{U_{t_n}}|
                \\& = \frac{|R_{U_{t_n}}|}{\frac{\theta}{\lambda}-1}+  \varepsilon_f+|R_{U_{t_n}}|
                \\& = (1+ \frac{1}{\frac{\theta}{\lambda}-1} )|R_{U_{t_n}}|+ \varepsilon_f
                \\&\leq \approxNT (1+ \frac{1}{\frac{\theta}{\lambda}-1} )\costopt + \varepsilon_f
            \end{align*}
                    
        In addition, since $\hat{t}_f = \frac{|R_{U_{t_n}}|}{\frac{\theta}{\lambda}-1}$ 
        and $\hat{t}_f \leq t_n$,
        the completion time of \ouralg,
        given by $t_n + |R_{U_{t_n}}|$, is no later than
        $\frac{\theta}{\lambda}t_n$. We thus obtain the bound:
            \begin{align*}
                \costalg & = t_n + |R_{U_{t_n}}|
                \leq \frac{\theta}{\lambda} t_n
                \leq \frac{\theta}{\lambda}\costopt
            \end{align*}
    \end{itemize}

The proof is complete. 
\end{proof}

\medskip

Next, 
we assume 
that 
the server is performing the schedule $R_{U_{\hat{t}_p}}$ at time~$t_n$, 
which is the schedule that immediately precedes the last schedule $R_{U_{t_n}}$.
Here, $\hat{t}_p$ denotes the start time of $R_{U_{\hat{t}_p}}$, which is also the decision time at which~\ouralg~determines 
when
to initiate $R_{U_{\hat{t}_p}}$.
Note that upon completing $R_{U_{\hat{t}_p}}$, the server proceeds to perform the final schedule $R_{U_{t_n}}$, and all requests in $U_{t_n}$ arrive strictly after time~$\hat{t}_p$.
\longdelete{
the server is performing the schedule $R_{U_{\hat{t}_p}}$ at time $t_n$,
where $\hat{t}_p$ represents 
the start time of schedule $R_{U_{\hat{t}_p}}$.
Note that 
after finishing $R_{U_{\hat{t}_p}}$, there will be the latest schedule $R_{U_{t_n}}$, and all the requests in $U_{t_n}$ arrive after time $\hat{t}_p$. 
}

\begin{lemma}\label{lem:working}
Suppose the server is working at time $t_n$. The competitive ratio of~\ouralg~is
 $\max\{ \lambda \theta + \frac{\lambda \theta \varepsilon_f}{\costopt}, \approxNT (1+ \frac{1}{\frac{\theta}{\lambda}-1} ), (\frac{\lambda\theta}{2}+\approxNT) \}$-consistent
 and $\max\{\approxNT (1+ \frac{1}{\lambda\theta-1} ),  (\frac{\theta}{2\lambda}+\approxNT)\}$-robust, where $\lambda \in (\frac{1}{\theta},1]$ 
and 
$\varepsilon_f = | \hat{t}_f - t_n|$.  
\end{lemma}

\begin{proof}
We divide the scenario into two cases depending on the status of the server after 
completing
$R_{U_{\hat{t}_p}}$: 
either the server immediately begins performing the schedule $R_{U_{t_n}}$, or it waits at the origin for 
a while
before doing so.
The outcomes of the two cases are summarized in Table \ref{tab:lem_working},
where the different gadgets used in each case by \ouralg~are analyzed.

\begin{table}[h] 
\centering
\caption{The two cases in Lemma \ref{lem:working} considering the server being working on $R_{U_{\hat{t}_p}}$ at time $t_n$}\label{tab:lem_working}
\begin{tabular}{|c|c|c|c|}
\hline
Case  & Scenario  & \ouralg's gadget & Competitive ratio\\
\hline     
\multirow{2}{*}{Case 1} & \multirow{2}{*}{$\hat{t}_f >  \hat{t}_p+|R_{U_{\hat{t}_p}}|$}
                                 & Late-start & $\min\{ \lambda A+ \frac{\lambda \theta \varepsilon_f}{\costopt}, B_{1/\lambda}\}$ \\\cline{3-4}
                                 &&Early-start & $B_\lambda$ \\\cline{1-1}\cline{2-4}
\multirow{2}{*}{Case 2} & \multirow{2}{*}{$\hat{t}_f \leq \hat{t}_p+|R_{U_{\hat{t}_p}}|$}
                                 & Late-start  & $C_\lambda$\\\cline{3-4}
                                 && Early-start & $C_{1/\lambda}$\\\cline{1-1}\cline{2-4}
\end{tabular}
\end{table}
    
    \begin{itemize}
                \item{Case 1: 
                Assume the server
                of \ouralg~waits at 
                the
                origin for a while after finishing $R_{U_{\hat{t}_p}}$. That is, $\hat{t}_f >  \hat{t}_p+|R_{U_{\hat{t}_p}}|$}.
                
                We then analyze the relationship between $\hat{t}_f$ and $t^\prime_f$, which are the decision times at which~\ouralg\ and \smartstart, respectively, determine 
                when
                to initiate the schedule $R_{U_{t_n}}$.
                \begin{itemize}
                    \item{1-1. $\hat{t}_f \geq t^\prime_f$, as shown in Fig.~\ref{case2-1-1}. }\\
                    \begin{figure}[htb] 
                        \centering 
                        \includegraphics[width=0.55\textwidth]{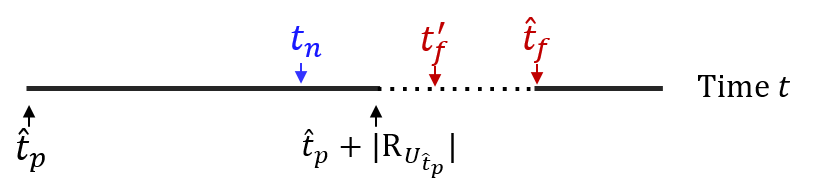}  
                        \caption{Case 1-1: $\hat{t}_f >  \hat{t}_p+|R_{U_{\hat{t}_p}}| $ and $\hat{t}_f \geq t^\prime_f$}
                        \label{case2-1-1}
                    \end{figure}
                    
                     The server of \ouralg~begins performing $R_{U_{t_n}}$ at $\hat{t}_f$. In addition, as $\hat{t}_f \geq t^\prime_f$ 
                     (i.e. the late-start case),
                     it satisfies $\hat{t}_f = \frac{|R_{U_{t_n}}|}{\lambda \theta-1}$.
                     Therefore, the completion time of \ouralg~is $\hat{t}_f + |R_{U_{t_n}}| =\lambda \theta \hat{t}_f $. We then have the bound: 
                     \begin{align*}
                        \costalg & = \hat{t}_f + |R_{U_{t_n}}|
                        = \lambda \theta \hat{t}_f 
                        = \lambda \theta (t_n+ \varepsilon_f)
                        \\& = \lambda \theta t_n + \lambda \theta \varepsilon_f 
                        \leq \lambda \theta \costopt + \lambda \theta \varepsilon_f
                    \end{align*}
                    We can
                    replace $\hat{t}_f$ with $\frac{|R_{U_{t_n}}|}{\lambda \theta-1}$ 
                    and 
                    obtain the bound without the waiting error $\varepsilon_f$:
                    \begin{align*}
                        \costalg & = \hat{t}_f + |R_{U_{t_n}}|
                        = \frac{|R_{U_{t_n}}|}{\lambda \theta-1} + |R_{U_{t_n}}|
                        \\& =(1+ \frac{1}{\lambda \theta-1})|R_{U_{t_n}}|
                        \\&\leq \approxNT (1+ \frac{1}{\lambda \theta-1})\costopt
                    \end{align*}
                    \item{1-2. $\hat{t}_f \leq t^\prime_f$, as shown in Fig.~\ref{case2-1-2}.}\\
                    The server of \ouralg~begins performing $R_{U_{t_n}}$ at $\hat{t}_f$,
                    which is earlier than the corresponding decision time of~\smartstart.
                        \begin{figure}[htb] 
                            \centering 
                            \includegraphics[width=0.55\textwidth]{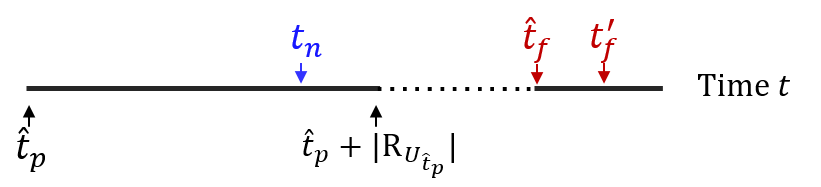}  
                            \caption{Case 1-2: $\hat{t}_f >  \hat{t}_p+|R_{U_{\hat{t}_p}}| $ and $\hat{t}_f \leq t^\prime_f$}
                            \label{case2-1-2}
                        \end{figure}
                        
                    The completion time of~\ouralg~is $\hat{t}_f + |R_{U_{t_n}}|$. In addition, as $ \hat{t}_f \leq t^\prime_f$ 
                    (i.e. the early-start case),
                    it satisfies $\hat{t}_f = \frac{|R_{U_{t_n}}|}{\frac{\theta}{\lambda}-1}$. We thus replace $\hat{t}_f$ with $\frac{|R_{U_{t_n}}|}{\frac{\theta}{\lambda}-1}$
                    and obtain the bound: 
                        \begin{align*}
                            \costalg & = \hat{t}_f + |R_{U_{t_n}}|
                            = \frac{|R_{U_{t_n}}|}{\frac{\theta}{\lambda}-1}+ |R_{U_{t_n}}| 
                            \\& =(1+ \frac{1}{\frac{\theta}{\lambda}-1} )|R_{U_{t_n}}|
                            \\&\leq \approxNT (1+ \frac{1}{\frac{\theta}{\lambda}-1} )\costopt.
                        \end{align*}
                \end{itemize}

                \longdelete{
                Case 2 presents that the server immediately starts the final schedule $R_{U_{t_n}}$ after completing $R_{U_{\hat{t}_p}}$. In the scenario, we show the bound of the final schedule $R_{U_{t_n}}$ in Lemma~\ref{lem:finalschedule}, which was mentioned in the proof of the \ignore~algorithm in \cite{ascheuer2000online}. Note that \ignore~is a non-zealous schedule-based algorithm. It implies that the server will start the schedule from the origin if there exists an unserved request. Therefore, we can apply the result to Case 2 of Lemma~\ref{lem:working}.

                \begin{restatable}[\cite{ascheuer2000online}, Theorem 4.]{rLem}{lemfr}\label{lem:finalschedule}
                    Let the last request be denoted by $x_n = (a_n,b_n,t_n)$ if the server is currently working on schedule $R_{U_{t_p}}$ at time $t_n$. Then, we can obtain:
                    \begin{center}
                        $|R_{U_{t_n}}| \leq \approxNT\costopt - \approxNT t_p + \approxNT d(o,a_j)$
                    \end{center}
                    We assume that the request $x_j = (a_j,b_j,t_j)$ in $U_{t_n}$ is the first request served during schedule $R_{U_{t_n}}$ determined by OPT. 
                \end{restatable}

                \begin{corollary}[\cite{ascheuer2000online}, Theorem 4.]\label{cor:optbound}
                    Regarding the requests $x_j = (a_j,b_j,t_j)$ and the future requests after $t_j$, the inequality is as follows: $\costopt \geq t_j + d(o,a_j) \geq t_p + d(o,a_j)$. 
                \end{corollary}
                }

                \item{Case 2: 
                Assume the server of~\ouralg~immediately 
                performs
                $R_{U_{t_n}}$ 
                upon completing
                $R_{U_{\hat{t}_p}}$, as shown in Fig.~\ref{case2-2}.
                That is, $\hat{t}_f \leq \hat{t}_p+|R_{U_{\hat{t}_p}}|$,
                and the completion time of~\ouralg~is $\hat{t}_p +|R_{U_{\hat{t}_p}}| + |R_{U_{t_n}}|$.}
                \begin{figure}[htb] 
                        \centering 
                        \includegraphics[width=0.6\textwidth]{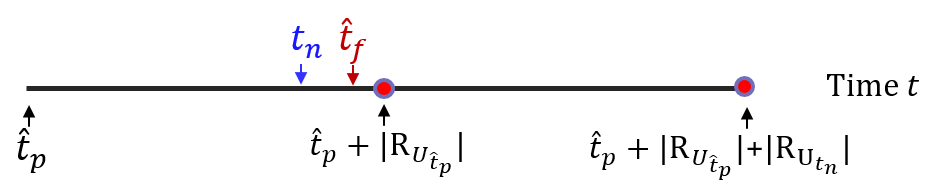} 
                        \caption{Case 2: $\hat{t}_f \leq \hat{t}_p+|R_{U_{\hat{t}_p}}| $
                        }
                        \label{case2-2}
                    \end{figure}
                
                Before getting into the analysis, 
                we refer to~\cite{ascheuer2000online} and follow the proof structure presented in Theorem 6.
                Moreover, $|R_{U_{t_n}}|$ can be bounded by $ \approxNT\costopt - \approxNT \hat{t}_p + \approxNT d(o,a_k)$, as stated in Lemma~\ref{lem:finalschedule}, based on Theorem 4 of the \ignore~algorithm in \cite{ascheuer2000online}. Note that \ignore~is a zealous schedule-based algorithm, implying that the server performs a schedule from the origin if there are unserved requests. We recall the previous result as follows: 
                    

                \begin{lemma}[\cite{ascheuer2000online}, Theorem 4.]\label{lem:finalschedule}
                    Let the last request be denoted by 
                    $x_n = (a_n,t_n)$,
                    and suppose it arrives while the server is performing the schedule $R_{U_{t_p}}$.
                    Note that $R_{U_{t_p}}$ was initiated at time $t_p$.
                    We can derive the following:
                    \begin{center}
                        $|R_{U_{t_n}}| \leq \approxNT \costopt - \approxNT t_p + \approxNT d(o,a_k)$
                    \end{center}
                    Let
                    $x_k = (a_k,t_k)$ be the first request in $U_{t_n}$ that is 
                    served 
                    in the schedule $R_{U_{t_n}}$ as determined by OPT.
                    For the request $x_k = (a_k,t_k)$ and the subsequent requests arriving after $t_k$, the following inequality holds: 
                    \begin{center}
                    $\costopt \geq t_k + d(o,a_k) \geq t_p + d(o,a_k)$
                    \end{center}
                \end{lemma}

                    \begin{itemize}
                        \item {2-1. 
                        If~\ouralg~performs
                        $R_{U_{\hat{t}_p}}$ later than~\smartstart~
                        (i.e. the late-start case),
                        it implies that the following conditions are satisfied:
                        $\hat{t}_p \geq \frac{|R_{U_{\hat{t}_p}}|}{\lambda\theta-1}$ and $\hat{t}_p +|R_{U_{\hat{t}_p}}| \leq \lambda\theta\hat{t}_p$.}
                        We then combine Lemma \ref{lem:finalschedule} and have the inequality: 
                        \begin{align*}
                            \costalg & = \hat{t}_p +|R_{U_{\hat{t}_p}}| + |R_{U_{t_n}}| 
                            \\&\leq \lambda\theta\hat{t}_p + \approxNT \costopt- \approxNT\hat{t}_p+\approxNT d(o,a_k)
                            \\&\leq(\lambda\theta-\approxNT)\hat{t}_p+\approxNT\costopt+\approxNT d(o,a_k).
                        \end{align*}

                         Based on another outcome presented in Lemma~\ref{lem:finalschedule}, it follows that $\costopt \geq \hat{t}_p + d(o,a_k)$.
                         Consequently, we have $\hat{t}_p \leq \costopt - d(o,a_k)$. Additionally, 
                         since $x_k = (a_k, t_k)$ represents an actual request and the server eventually needs to return to the origin, the distance $d(o, a_k)$ 
                         is
                         bounded by $\costopt/2$. We thus obtain the following bound:
                        \begin{align*}
                            \costalg &\leq(\lambda\theta-\approxNT)(\costopt-d(o,a_k))+\approxNT\costopt+\approxNT d(o,a_k)
                            \\& = \lambda\theta\costopt+(2\approxNT-\lambda\theta)d(o,a_k)
                            \\& \leq \lambda\theta\costopt + (\approxNT-\frac{\lambda\theta}{2})\costopt
                            \\& = (\frac{\lambda\theta}{2}+\approxNT)\costopt.
                        \end{align*}
                        \item {2-2. 
                        If~\ouralg~performs $R_{U_{\hat{t}_p}}$ earlier than~\smartstart~
                        (i.e. the early-start case),
                        it implies that the following conditions are satisfied: $\hat{t}_p \geq \frac{|R_{U_{\hat{t}_p}}|}{\frac{\theta}{\lambda}-1}$ and $\hat{t}_p +|R_{U_{\hat{t}_p}}| \leq \frac{\theta}{\lambda}\hat{t}_p$.}
                        Therefore, we apply a similar argument in the previous case (Case 2-1) and have the bound: 
                        \begin{align*}
                            \costalg & = \hat{t}_p +|R_{U_{\hat{t}_p}}| + |R_{U_{t_n}}|
                            \\&\leq \frac{\theta}{\lambda}\hat{t}_p + \approxNT \costopt- \approxNT\hat{t}_p+\approxNT d(o,a_k)
                            \\&\leq(\frac{\theta}{\lambda}-\approxNT)\hat{t}_p+\approxNT\costopt+\approxNT d(o,a_k)
                            \\&\leq(\frac{\theta}{\lambda}-\approxNT)(\costopt-d(o,a_k))+\approxNT\costopt+\approxNT d(o,a_k)
                            \\&\leq(\frac{\theta}{2\lambda}+\approxNT)\costopt
                        \end{align*}
                \end{itemize}   
            \end{itemize}   
The proof is complete. 
\end{proof} 

\medskip

Finally, combining the results 
from
Lemma~\ref{lem:waiting} and Lemma~\ref{lem:working}, we 
establish
the consistency bound 
as
$\max\{\lambda \theta+\frac{\lambda \theta \varepsilon_f}{\costopt}, \approxNT (1+ \frac{1}{\frac{\theta}{\lambda}-1} ) + \frac{\varepsilon_f}{\costopt}, (\frac{\lambda\theta}{2}+\approxNT) \} $
and the robustness bound 
as
$\max\{ \frac{\theta}{\lambda},\approxNT (1+ \frac{1}{\lambda\theta-1} ),  (\frac{\theta}{2\lambda}+\approxNT)\} $.

\section{Conclusion}

In this study, we have exploited learning-augmented techniques to develop the \ouralg~algorithm for the OLTSP and the OLDARP, 
which shows its consistency closely approaches to our proposed lower bound of 2. 
In particular, we have applied the concept of online predictions to determine when to perform every schedule in a schedule-based algorithm. 
Precisely, we simply make a binary decision, choosing whether the server should wait longer at the origin or perform the schedule early. 
Note that \ouralg~using online predictions does not need to know the total number of requests in advance, which is different from most of the previous studies in the learning-augmented framework. When the number of requests is known, 
the naive online algorithm by Nagamochi et al.~\cite{nagamochi1997complexity} showed the upper bound of 2.5, while Yu et al.~\cite{zhang} proved the lower bound of 2.5 for any schedule-based algorithms if the Christofides’ heuristic is used.
We also remark that online predictions over a waiting strategy for scheduling problems, e.g., job-shop scheduling, might be of independent interest. It would be worthwhile to extend the notion to other problems in an online fashion. 

\newpage
\nocite{*}
\bibliography{ref}
\bibliographystyle{plain}

\newpage

\appendix

\end{document}